\documentclass[a4paper]{article} % Aptara syntax

\usepackage[ruled]{algorithm2e}

\SetAlFnt{\small}
\SetAlCapFnt{\small}
\SetAlCapNameFnt{\small}
\SetAlCapHSkip{0pt}
\IncMargin{-\parindent}

\usepackage{amsthm,amssymb,mathtools,stmaryrd,tikz,subfigure,hyperref}

\usepackage{fullmacros}

\newcommand{\altpath}[1]{\mathrm{AltPath}(#1)}
\newcommand{\seq}[1]{\mathbf{#1}}
\newcommand{\pushed}[2][n]{{\tt push}^{#1}_{#2}}
\newcommand{\pair}{{\tt pair}}

\renewcommand{\word}[1]{{\mathtt #1}}

\renewcommand{\twwprojects}{\cond{\sharp W}_{2}}
\renewcommand{\graphtw}[1]{W_{\word{#1}}}
\renewcommand{\graphingtw}[1]{W_{\word{#1}}}
\renewcommand{\graphtw}[1]{\bar{W}_{\word{#1}}}
\renewcommand{\ListType}{\cond{Nat}_{2}}
\renewcommand{\reptw}[1]{{\tt Gp}(\word{#1})}

\renewcommand{\repany}[1]{\mathbf{Rep}(\word{#1})}

\newcommand{\pred}[1]{\mathbf{Pred}(\microcosm{#1})}
\newcommand{\ListTypeELL}{{\rm BList}}

\newcommand{\measbig}[1]{\left\llbracket #1 \right\rrbracket}

\newcommand{\lin}{{\tt inl}}
\newcommand{\rin}{{\tt inr}}

\newcommand{\icc}{{\sc icc}\xspace}
\newcommand{\GoI}{{\sc goi}\xspace}
\newcommand{\masas}{{\sc masas}\xspace}

\begin{document}

\title{Interaction Graphs: Nondeterministic Automata}
\author{Thomas Seiller \\University of Copenhagen}
% NOTE! Affiliations placed here should be for the institution where the
%       BULK of the research was done. If the author has gone to a new
%       institution, before publication, the (above) affiliation should NOT be changed.
%       The authors 'current' address may be given in the "Author's addresses:" block (below).
%       So for example, Mr. Abdelzaher, the bulk of the research was done at UIUC, and he is
%       currently affiliated with NASA.

\maketitle

\begin{abstract}
 This paper exhibits a series of \emph{semantic} characterisations of sublinear nondeterministic complexity classes. These results fall into the general domain of logic-based approaches to complexity theory and so-called \emph{implicit computational complexity} (\icc), i.e.\ descriptions of complexity classes without reference to specific machine models. In particular, it relates strongly to \icc results based on linear logic since the semantic framework considered stems from work on the latter. Moreover, the obtained characterisations are of a geometric nature: each class is characterised by a specific action of a group by measure-preserving maps.
\end{abstract}

\section{Introduction}

 %This paper provides new characterisations of complexity classes based on \emph{complexity-through-realisability} techniques described in the perspective paper \cite{seiller-towards}. One of the final aims of this research program is to enable the use of techniques from the mathematical theories of operator algebras, dynamical systems and ergodic theory against open problems in complexity theory. We also believe that the \emph{complexity-through-realisability} theory provides well-grounded mathematical foundations for complexity theory, one that is machine-independent and applies to different computational paradigms. 
 
 %It moreover lays ground to the general theory and provide the formal background and definitions . 
% The author tried as much as possible to express these definitions in all their generality, so that they may be applied without being modified to study different computational paradigms such as quantum computation. %providing some intuitions on how these methods will apply to deal with not only sequential computation but other computational paradigms such as probabilistic and quantum computation. 
% The generality of the approach is illustrated by the obtention of first implicit characterisations of an infinite family of probabilistic complexity class, between stochatic languages \Stochastic and the class \PLogspace of (unbounded error) probabilistic logarithmic space. The interested reader can consult the complementary perspective paper \cite{seiller-towards} for an overview of the motivations and the general spirit behind this work.

Complexity theory is concerned with the study of how many resources are needed to perform a specific computation or to solve a given problem. The study of \emph{complexity classes} -- sets of problems which need a comparable amount of resources to be solved, lies at the intersection of mathematics and computer science. After the obtention of strong impossibility results \cite{naturalproofs} preventing the use of known proof methods to settle open separation problems, mathematicians have tried to give characterisations of complexity classes that differ from the original machine-bound definitions, hoping to enable methods from radically different areas of mathematics. % to be used against open problems.

Among them, the field of Implicit Computational Complexity (\icc) aims at studying algorithmic complexity only in terms of restrictions of languages and computational principles.
It has been established since Bellantoni and Cook' landmark paper \cite{bellantonicook}, and following work by Leivant and Marion \cite{leivantmarion1,leivantmarion2}.
Amongst the different approaches to \icc, several results were obtained by considering syntactic restrictions of \emph{linear logic} \cite{ll}, a refinement of intuitionnistic logic which accounts for the notion of resources. Linear logic introduces a modality $\oc$ marking the \enquote{possibility of duplicating} a formula $A$: the formula $A$ shall be used exactly once, while the formula $\oc A$ can be used any number of times. Modifying the rules governing this modality then yields variants of linear logic having computational interest: this is how constrained linear logic systems, for instance \BLL \cite{BLL} and \ELL \cite{danosjoinet}, are obtained.

Recently, a new line of research emerged, providing semantic characterisations of complexity classes instead of syntactical ones. 
%Previous related work began in collaboration with C. Aubert  \cite{seiller-conl,seiller-lsp} and provided new characterisations of logarithmic space complexity classes as sets of operators in a von Neumann algebra. 
This approach was initiated by Girard  \cite{normativity} and motivated by his work on Geometry of Interaction (\GoI) models, and more precisely the hyperfinit \GoI model \cite{goi5}. Together with C. Aubert, the author showed how Girard's proposal lead to the characterisation of \coNLogspace%\footnote{Although the Immerman-Szelepcsenyi theorem \cite{szelepcsenyi,immerman} states that \coNLogspace equals \NLogspace, the mentioned works willingly refrained from using it to keep alive the possibility of finding an alternative proof.}
 \cite{seiller-phd,seiller-conl} and \Logspace  \cite{seiller-lsp}. Unfortunately, technical reasons lead the authors to consider modifications of the initial hyperfinite \GoI framework, furthering characterisation results from the \GoI models construction. In other words, although originating from considerations on semantics, these results were not directly logic-related.

These semantic results were then rephrased in more syntactic terms, providing new characterisations related to logic programming results \cite{seiller-LPLS,seiller-ptime} but taking another step further from the initial framework of the hyperfinite \GoI model. After a first step which ended in the loss of an underlying logical framework, this second step ended in the loss of the rich mathematical theories the method was initially founded upon. 
Although this recent line of work have its own interests, it is the author's belief that one should not forget the mathematical structure from which these characterisations originated. 
This sentiment is strengthened by the author's discovery of a correspondence between fragments of linear logic and a classification of maximal abelian subalgebras (\masas) of von Neumann algebras \cite{seiller-masas}. %Although a more complete correspondence using the operator algebraic approach seems difficult due to the lack of intuitions (how does a subalgebra provide restrictions on computations?) and the myriad of still-open problems concerning \masas, the use of the author's interaction graphs models can provide alternative methods to . Indeed, graphings can be understood as generalisations of operators in a von Neumann algebra (graphings allow for divergence, while operators in a von Neumann algebra are \emph{bounded}), and the work presented here should be considered as a . %More precisely graphings are defined as (equivalence classes of) families of restrictions of measurable maps -- i.e. \emph{edges} -- in a monoid $\microcosm{m}$ of measurable maps from a measurable space $\measure{X}$ to itself; in the general setting these \emph{edges} may be given a weight in an arbitrary monoid $\Omega$ which we assume for this discussion to be a submonoid of the complex numbers. 
The approach taken in this paper is therefore quite orthogonal to the recent evolutions of the subject, as it aims at the obtention of a deeper understanding of how complexity classes can be related to the mathematics behind \GoI models in order to provide complexity theorists with new techniques and invariants \cite{seiller-towards}.

\subsection*{Contributions and Outline}

The present work achieves three distinct goals related to the logic-based characterisations of complexity classes. Firstly, complexity classes are here characterised as specific types in models of (fragments of) linear logic. It thus fills the gap between the above mentioned series of work \GoI-inspired results in computational complexity \cite{seiller-conl,seiller-lsp,seiller-LPLS,seiller-ptime} and the actual semantics provided by \GoI models. Secondly, we obtain characterisations of several classes that were not available using previous techniques. This is due to a change of perspective which allows new proof techniques, sensible to more subtle differences. Thirdly, each complexity class considered is here characterised by a specific \emph{group action} on a measure space. This hints at possible uses of mathematical invariants from ergodic theory and measurable group theory in the context of computational complexity.

The paper is constructed as follows. The next section introduces the technical material about \emph{interaction graphs models of linear logic}. This will allow us to define, in section 3, the ambient model which will be used to obtain the characterisations. We also define the representation of binary words and the notion of $\microcosm{m}_{k}$-machine. Section 4 contains the technical proof of the characterisation: after recalling the definition of multihead automata, we show how the complexity class captured by $k$-head multihead automata and the one captured by our notion of $\microcosm{m}_{k}$-machine coincide. Lastly, we discuss this result in the conclusion, providing both a logic and a geometric reading of it.

\section{Interaction Graphs Models}

\subsection{Basic Definitions}\label{subsec:basics}

Interaction graphs models were introduced by the author in a series of papers \cite{seiller-goim,seiller-goiadd,seiller-goig,seiller-goie,seiller-goif}. It is a modular framework providing a rich hierarchy of models of (fragments of) linear logic. We describe here the basic operations needed to work out the following section. Proofs are interpreted as a generalisation of graphs, named \emph{graphings}. Graphings can be understood as graphs \emph{realised} on a measured space, i.e. vertices are measurable subsets of the space, and edges represents measurable functions mapping the source subset to the target subset. As part of the modularity of the framework, we use the notion of \emph{microcosm} to restrict the set of measurable maps the edges of the graphing considered can represent.

For technical reasons explained in earlier papers \cite{seiller-goig}, all measurable maps cannot be used to represent edges. To be able to define models of linear logic, one has to restrict to non-singular measurable-preserving maps. We recall that a map $f$ is \emph{non-singular transformation} if it is a measurable map such that $\mu(f(A))=0$ if and only if $\mu(A)=0$. We say $f$ is \emph{measurable-preserving} when $f(A)\in\tribu{B}$ whenever $A\in\tribu{B}$. 

\begin{definition}[Microcosm]\label{microcosm}
Given a measure space $\measure{X}=(\space{X},\tribu{B},\mu)$, a \emph{microcosm} on $\measure{X}$ is a set $\microcosm{m}$ of non-singular measurable-preserving transformations $\measure{X}\rightarrow\measure{X}$ which has the structure of a monoid w.r.t. the composition of maps.
\end{definition}

In practice, we define microcosms by providing a set of generating maps; this defines a unique microcosm, namely the smallest microcosm containing all given maps. %The latter can be defined formally as the set of all finite compositions of elements of $G$ quotiented by every equality everywhere.

\begin{examples}\label{examplemicrocosms}
For all examples considered in this section, we will restrict to the underlying measure space the real line $\realN$ endowed with the Lebesgue measure. We first define the microcosm $\microcosm{z}$ as the set of all integral translations on $\realN$, i.e. 
$$\microcosm{z}=\{T_{k}: \realN\rightarrow\realN, x\mapsto x+k~|~ k\in\integerN \}.$$
Notice that this microcosm is generated by the set $\{T_{1},T_{-1}\}$.

Now, we can also define the microcosm $\microcosm{h}$ of integral homotheties on $\realN$, i.e. 
$$\microcosm{h}=\{H_{z}: \realN\rightarrow\realN, x\mapsto z.x~|~ z\in\integerN \}.$$
For this microcosm, no finite generating set exists. The following (infinite) set is however generating: $\{H_{p}~|~ (-p)\text{ is prime or equal to $1$}\}$.

These two microcosms are almost disjoint, as only the identity map on $\realN$ belongs to both of them. They are however submonoids of several common microcosms; in particular there exists a minimal such microcosm, namely the monoid of all integral affine transformations, i.e.
$$\microcosm{aff}=\{A_{k,h}: \realN\rightarrow\realN, x\mapsto h.x+k~|~ k,h\in\integerN \}.$$

Finally, all microcosms on a measure space $\measure{X}$ are submonoids of \emph{the largest microcosm on $\measure{X}$} -- called the \emph{macrocosm} -- defined as the set of all non-singular measurable-preserving transformations on $\measure{X}$.
\end{examples}

We must point out that a more general notion of microcosm was introduced in a recent work by the author \cite{seiller-goif}; the restricted notion defined here is however easier to grasp and sufficient for our purposes in this paper. We now define the notion of graphing. 

\begin{definition}[Graphing representative]
We fix a measure space $\measure{X}$, a microcosm $\microcosm{m}$, a monoid $\Omega$, a measurable subset $V^{G}$ of $\measure{X}$, and a finite set $D^{G}$. A ($\Omega$-weighted) \emph{$\microcosm{m}$-graphing representative} $G$ of \emph{support} $V^{G}$ and \emph{dialect} $D^{G}$ is a countable set 
$$\{(S^{G}_{e},\ttrm{i}^{G}_{e},\ttrm{o}^{G}_{e},\phi^{G}_{e},\omega^{G}_{e})~|~e\in E^{G}\},$$
where $S^{G}_{e}$ is a measurable subset of $V^{G}\times D^{G}$, $\phi^{G}_{e}$ is an element of $\microcosm{m}$ such that $\phi_{e}^{G}(S^{G}_{e})\in V^{G}$, $\ttrm{i}^{G}_{e},\ttrm{o}^{G}_{e}$ are elements of $D^{G}$, and $\omega^{G}_{e}\in\Omega$ is a \emph{weight}. We will refer to the indexing set $E^{G}$ as the set of \emph{edges}. For each edge $e\in E^{G}$ the set $S^{G}_{e}\times\{\ttrm{i}^{G}_{e}\}$ is called the source of $e$, and we define the \emph{target} of $e$ as the set $T^{G}_{e}\times\{\ttrm{o}^{G}_{e}\}$ where $T^{G}_{e}=\phi^{G}_{e}(S^{G}_{e})$.
\end{definition}

To provide some intuitions, we first ignore the dialect $D^{G}$, or equivalently we consider $D^{G}$ to be a singleton. Given an edge $e\in E^{G}$, the intuition is that the triple $(S^{G}_{e},\phi^{G}_{e},\omega^{G}_{e})$ corresponds to the following information: the source $S^{G}_{e}$ of the edge, the target $T^{G}_{e}:=\phi^{G}_{e}(S^{G}_{e})$ of the edge, the weight $\omega^{G}_{e}$ of the edge. Consequently, a graphing may be mapped to a graph whose edges are measurable subsets of $\measure{X}$. However, two different graphings may give rise to the same graph, as this mapping forgets about \emph{how} the source is mapped to the target, i.e. which measurable map in the microcosm realises the edge. 

The additional information of the elements $\ttrm{i}^{G}_{e},\ttrm{o}^{G}_{e}$ corresponds to \emph{control states}. Indeed, thinking of the finite set $D^{G}$ as a set of control state is a good intuition that can be followed through this paper. Building on this, one can define a weighted automata from a graphing as follows: the automata works on the (infinite) input alphabet consisting of all measurable subsets of $\measure{X}$ and has as set of states $D^{G}$; then each edge $e$ defines a transition from $S^{G}_{e}$ in state $\ttrm{i}^{G}_{e}$ to $T^{G}_{e}$ in state $\ttrm{o}^{G}_{e}$. This mapping, however, is again non-injective as it does not account for \emph{how} the source is mapped to the target.

\begin{examples}\label{examplegraphingsrep}
We first consider an example of \emph{deterministic graphing representative}, i.e. one such that every $x\in \measure{X}$ belongs to the source of at most one edge (up to a null measure set). For the sake of simplicity, the graphing representatives $F,G$ we consider are such that $D^{F}=D^{G}=\{\star\}$, i.e. they have a unique control state, and all weights will be equal to $1$; they are then defined by $V^{F}=V^{G}=[0,2[$ and 
$$ F=\{([0,1[,\star,\star,x\mapsto x+1,1),([1,2[,\star,\star,x\mapsto x-1,1)\} $$
$$ G=\{([0,1[,\star,\star,x\mapsto x+1,1),([1,2[,\star,\star,x\mapsto 2-x,1)\} $$
Note that these two examples give rise to the same graph and the same automata through the mapping just explained above. They are however quite different. In particular, using the notations of \autoref{examplemicrocosms}, the graphing $F$ is a $\microcosm{t}$-graphing while $G$ is not. Indeed, $G$ is neither a $\microcosm{t}$-graphing or a $\microcosm{h}$-graphing; it is however a $\microcosm{aff}$-graphing.
\end{examples}

Even though the intuitions given above are good to keep in mind, they are only approximations of the actual notion of graphing. Indeed, a graphing is defined as an equivalence class of graphing representatives. In particular, a graphing is not a specific set of edges realised by elements of a given microcosms: it is the generalised measurable dynamical system underlying a specific representation. In particular, both intuitions of graphings as graphs and automata fail to convey this idea that we now illustrate on an example.

\begin{examples}\label{examplegraphings}
We consider the graphing $F$ defined in \autoref{examplegraphingsrep}. We define the graphing $H$ defined by $D^{H}=\{\star\}$, $V^{H}=[0,2[$, and
$$ H = \{(]0,1/2[,\star,\star,x\mapsto x+1,1),(]1/2,1[,\star,\star,x\mapsto x+1,1),([1,2[,\star,\star,x\mapsto x-1,1)\} $$
The notion of graphing should be such that $F$ and $H$ are representative of the same graphing. To understand this, consider the graphing $H'$  defined by $D^{H'}=\{\star\}$, $V^{H'}=[0,2[$, and
$$ H' = \{([0,1/2[,\star,\star,x\mapsto x+1,1),([1/2,1[,\star,\star,x\mapsto x+1,1),([1,2[,\star,\star,x\mapsto x-1,1)\} $$
Then $H'$ is a \emph{refinement} of $H$ in that we only replaced the edge $([0,1[,\star,\star,x\mapsto x+1,1)$ by the two edges $([0,1/2[,\star,\star,x\mapsto x+1,1)$ and $([1/2,1[,\star,\star,x\mapsto x+1,1)$ to define $H'$ from $F$. Moreover, $H'$ is \emph{almost-everywhere equal} to $H$.
\end{examples}

As illustrated by the example, it is natural to identify graphing representatives w.r.t. almost-everywhere equality and a notion of \emph{refinement}, both combined in the following formal definition which is studied in earlier work \cite{seiller-goig}. 

\begin{definition}
A graphing representative $F$ is a refinement of a graphing representative $G$ if there exists a partition\footnote{We allow the sets $E^{F}_{e}$ to be empty.} $(E^{F}_{e})_{e\in E^{G}}$ of $E^{F}$ such that: 
$$
\forall e\in E^{G}, \cup_{f\in E^{F}_{e}} S^{F}_{f} =_{a.e.} S^{G}_{e};
\hspace{1cm} 
\forall e\in E^{G}, \forall f\neq f'\in E^{F}_{e},  \mu(S^{F}_{f} \cap  S^{F}_{f'})=0;
$$
$$
\forall e\in E^{G}, \forall f\in E^{F}_{e}, \omega^{F}_{f}=\omega_{e}^{G}
\hspace{1cm} 
\forall e\in E^{G}, \forall f\in E^{F}_{e}, \phi^{F}_{f}=\phi_{e}^{G}
$$
\end{definition}

Then two graphing representatives are \emph{equivalent} if and only if they possess a common refinement. The actual notion of \emph{graphing} is then an equivalence class of the objects just defined w.r.t. this equivalence. Since all operations considered on graphings were shown to be compatible with this quotienting \cite{seiller-goig}, i.e.\ well defined on the equivalence classes, we will in general make no distinction between a graphing -- as an equivalence class -- and a graphing representative belonging to this equivalence class.

\subsection{Paths and Execution}\label{subsec:execution}

In previous work, the author showed how to build denotational models of \emph{types}, or \emph{formulas}, by using graphings (over a space $\measure{X}$ chosen once and for all) to interpret \emph{programs}, or \emph{proofs}, depending on which side of the proofs-as-program correspondence we are standing on. These denotational models should be described as \emph{dynamic}, as they represent program execution, or the cut-elimination procedure, as a non-trivial operation in the semantics. In that aspect, they are distinguished from so-called \emph{static} denotational models in which a proof and its normal form have the same \enquote{denotation}. In the specific models built from graphings, the dynamic aspect is represented by the operation of \emph{execution}, based on the computation of alternating paths.

An \emph{alternating} path between two $\microcosm{m}$-graphings $F, G$ is a sequence of edges $\pi=e_{1}, e_{2},\dots, e_{k}$ verifying the following two conditions:
\begin{itemize}[nolistsep,noitemsep]
\item $e_{i}$ in $E^{G}$ if and only if $e_{i+1}\in E^{F}$, and $\ttrm{o}_{e_{i}}=\ttrm{i}_{e_{i+2}}$;
\item every measurable set $(\phi_{e_{i}}\circ \phi_{e_{i-1}}\circ \dots \phi_{e_{1}})(S_{e_{1}})$ is of strictly positive measure\footnote{In particular, $S_{e_{i+1}}\cap T_{e_{i}}$ is non-negligible.}.
\end{itemize} 
We denote $\altpath{F,G}$ the set of such paths. A given path naturally represents the composition $\phi_{\pi}=\phi_{e_{k}}\circ\dots\circ\phi_{e_{1}}$ which belongs to $\microcosm{m}$ since the latter is a monoid. We define the source of $\pi$ as $S_{\pi}\times\{\ttrm{i}_{e_{1}},\ttrm{i}_{e_{2}}\}$, where $S_{\pi}$ is defined as the set of all $x$ such that for all $i$, $\phi_{e_{i}}\circ\dots\circ\phi_{e_{1}}(x)\in S_{e_{i+1}}$. The weight $\omega_{\pi}$ of the path is defined as $\prod_{i=1}^{k}\omega_{e_{i}}$.

Given a path $\pi$ and a measurable subset $C$, we define $[\pi]_{o}^{o}(C)$ as the path representing the same map as $\pi$, and whose source has been restricted to $S^{\decoupe C}\times\{\ttrm{i}_{e_{1}},\ttrm{i}_{e_{2}}\}$, with $S_{\pi}^{\decoupe C}=(S_{\pi}\cap \bar{C}\cap (\phi_{\pi})^{-1}(\bar{C}))$ where $\bar{C}$ is the complement set of $C$. Intuitively, we restrict $\pi$ to the subset of its domain that lies outside of $C$ and is mapped outside of $C$ by the map $\phi_{\pi}$. The execution between graphings $F$, $G$ of respective supports $V+C$ and $W+C$ is then defined as the graphing $F\plug G$ of support $V+W$ consisting of all $[\pi]_{o}^{o}(C)$ for $\pi$ an alternating path between $F$ and $G$.

\begin{definition}[Execution]
Let $F$ and $G$ be graphings such that $V^{F}=V\disjun C$ and $V^{G}=C\disjun W$ with $V\cap W$ of null measure. Their \emph{execution} $F\plug G$ is the graphing of support $V\disjun W$ and dialect $D^{F}\times D^{G}$ defined as the set of all $[\pi]_{o}^{o}(C)$ where $\pi$ is an alternating path between $F$ and $G$.
$$ F\plug G=\{(S^{\decoupe C}_{\pi},(\ttrm{i}_{e_{1}},\ttrm{i}_{e_{2}}),(\ttrm{o}_{e_{n-1}},\ttrm{o}_{e_{n}}),\phi_{\pi},\omega_{\pi})~|~\pi= e_{1},e_{2},\dots,e_{n}\in\altpath{F,G}\} $$
\end{definition}

\begin{examples}
Consider the two graphings $F$ and $G$ shown in \autoref{figureexec} ($F$ is shown at the top of the figure; $G$ at the bottom). Their execution is then the graphing with the following countably infinite family of paths $\{a(db)^{k}ec\}_{k=0}^{\infty}$, where $a(db)^{k}ec$ is of source $[(2^{k-1}-1)/2^{k-1},(2^{k}-1)/2^{k}]$ and realised by the function $x\mapsto 2^{k}x-2^{k}+6$.
\end{examples}

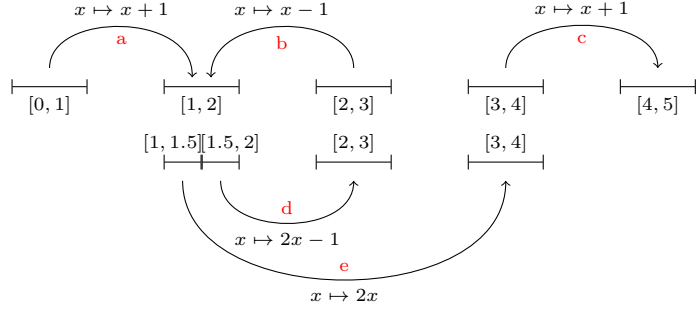
\begin{figure}
\centering
\begin{tikzpicture}
	\draw[|-|] (0,0) -- (1,0) node [midway,below] {\scriptsize{$[0,1]$}} node [midway,above] (1) {};
	\draw[|-|] (2,0) -- (3,0) node [midway,below] {\scriptsize{$[1,2]$}} node [midway,above] (2) {};
	\draw[|-|] (4,0) -- (5,0) node [midway,below] {\scriptsize{$[2,3]$}} node [midway,above] (3) {};
	\draw[|-|] (6,0) -- (7,0) node [midway,below] {\scriptsize{$[3,4]$}} node [midway,above] (4) {};
	\draw[|-|] (8,0) -- (9,0) node [midway,below] {\scriptsize{$[4,5]$}} node [midway,above] (5) {};
	
	\draw[->] (1) .. controls (0.5,1) and (2.4,1) .. (2.west) node [midway,below] {\scriptsize{\textcolor{red}{a}}} node [midway,above] {\scriptsize{$x\mapsto x+1$}};
	\draw[->] (3) .. controls (4.5,1) and (2.6,1) .. (2.east) node [midway,below] {\scriptsize{\textcolor{red}{b}}} node [midway,above] {\scriptsize{$x\mapsto x-1$}};
	\draw[->] (4) .. controls (6.5,1) and (8.5,1) .. (5) node [midway,below] {\scriptsize{\textcolor{red}{c}}} node [midway,above] {\scriptsize{$x\mapsto x+1$}};

	\draw[|-|] (2,-1) -- (2.5,-1) node [near start,above] {\scriptsize{$[1,1.5]$}} node [midway,below] (-2) {};
	\draw[|-|] (2.5,-1) -- (3,-1) node [near end,above] {\scriptsize{$[1.5,2]$}} node [midway,below] (-25) {};
	\draw[|-|] (4,-1) -- (5,-1) node [midway,above] {\scriptsize{$[2,3]$}} node [midway,below] (-3) {};
	\draw[|-|] (6,-1) -- (7,-1) node [midway,above] {\scriptsize{$[3,4]$}} node [midway,below] (-4) {};
	
	\draw[->] (-25) .. controls (2.75,-2) and (4.5,-2) .. (-3) node [midway,above] {\scriptsize{\textcolor{red}{d}}} node [midway,below] {\scriptsize{$x\mapsto 2x-1$}};
	\draw[->] (-2) .. controls (2.25,-3) and (6.5,-3) .. (-4) node [midway,above] {\scriptsize{\textcolor{red}{e}}} node [midway,below] {\scriptsize{$x\mapsto 2x$}};

\end{tikzpicture}
\caption{Two graphings F (above) and G (below).\label{figureexec}}
\end{figure}

Execution represents the cut-elimination procedure or, through the proofs-as-programs correspondence, the execution of programs. Together, graphings and execution provide dynamic semantics for proofs / programs. 

\subsection{Orthogonality and Models}

The second step in defining interaction graphs models consists in building an interpretation of types as (particular) sets of proof interpretations. This construction builds on a particular case of the (tight, orthogonality) double-glueing construction defined by Hyland and Schalk \cite{doubleglueing}. 

%\textcolor{red}{Discuss the difference with the game semantics apporach}
%This is a fundamental difference with the game-semantics approach. Indeed, the naive games and strategies discussed in the previous section provide degenerate models of linear logic, e.g. for multiplicative connectives the *-autonomous structure will be that of a compact-closed category. It is therefore natural to try and refine this model to obtain a non-degenerate model. A syntactic approach would be to define notions of non-naive games and strategies, obtaining a measurable version of usual game semantics.
%\textcolor{red}{Discuss the difference with the game semantics approach}

%First, one defines a notion of \emph{orthogonality}, i.e. a binary relation between proof representations (here, graphings). Then, one defines types as sets of such proof representation which are closed under taking the double-orthogonal (\ref{def:doubleorth}). In interaction graphs, an additional (technical) condition is required from these sets in order to obtain satisfying models of additives \cite{seiller-goiadd}. This condition -- closure under inflation (\ref{def:inflation}) -- will be introduced later.

We first define the \emph{measurement} between two graphings, and then use it to define a binary relation between graphings -- the \emph{orthogonality}. Although the definition of measurement is quite involved in the general case \cite{seiller-goig}, it will be enough for our purpose to consider:
\begin{itemize}[nolistsep,noitemsep] 
\item its restriction to measure-preserving maps; 
\item a fixed parameter map $m:\Omega\rightarrow \realposN\cup\{\infty\}$;
\end{itemize}
The measurement is defined by a sum over all \emph{circuits} between two graphs $F$ and $G$. A \emph{circuit} is an equivalence class of cycles w.r.t. cyclic permutations. The actual sum is computed by a choice of representatives of circuits, i.e. a set $\repcirc{F,G}$ of cycles $\pi=(e_{i})_{i=0}^{n-1}$ such that (1) $\pi^{k}\in\repcirc{F,G}$ ($k$ a non-zero integer) implies $\pi\in\repcirc{F,G}$; (2) $\pi\in\repcirc{F,G}$ implies $(e_{i+k})_{i=0}^{n-1}\not\in\repcirc{F,G}$ (where $i+k$ is computed in $\integerN/ n\integerN$). More details about the definition, and a proof that the considered measurement is independent from this choice of representative is found in previous work by the author \cite{seiller-goig}.

\begin{definition}
The measurement between two graphings (realised by measure-preserving maps) is defined as 
$$\meas[]{F,G}=\sum_{\pi\in\repcirc{F,G}} \int_{\supp{\pi}} \frac{m(\omega(\pi)^{\rho_{\phi_{\pi}}(x))}}{\rho_{\phi_{\pi}}(x)}d\lambda(x)$$
where $\rho_{\phi_{\pi}(x)}=\inf\{n\in\naturalN~|~ \phi_{\pi}^{n}(x)=x\}$ (by convention, $\inf\emptyset=\infty$).
\end{definition}

%The following simple lemma will be useful to simplify the computation of the measurement.
%
%\begin{lemma}
%If the graphings considered are such that all cycle $\pi$ verifies that $x\in\supp{\pi}\Rightarrow\phi_{\pi}(x)=x$, the measurement is computed as:
%$$\meas[]{F,G}=\sum_{\pi} \mu(\supp{\pi})\omega(\pi)$$
%\end{lemma}

We now describe the models. For technical reasons explained in previous papers \cite{seiller-goim,seiller-goiadd}, a proof is interpreted as a pair of a real number and a formal weighted sum of graphings of a fixed support -- a \emph{sliced graphing}. The measurement and the execution are extended to these objects as follows:
$$\measbig{(a,\sum_{i\in I}\alpha_{i}A_{i}),(b,\sum_{j\in J}\beta_{j}B_{j})}=a\left(\sum_{j\in J}\beta_{j}\right)+b\left(\sum_{i\in I}\alpha_{i}\right)+\sum_{(i,j)\in I\times J}\alpha_{i}\beta_{j}\meas[]{A_{i},B_{j}}$$
$$(a,\sum_{i\in I}\alpha_{i}A_{i})\plug(b,\sum_{j\in J}\beta_{j}B_{j})=\left(\measbig{(a,\sum_{i\in I}\alpha_{i}A_{i}),(b,\sum_{j\in J}\beta_{j}B_{j})},\sum_{(i,j)\in I\times J}\alpha_{i}\beta_{j}A_{i}\plug B_{j}\right)$$

\begin{definition}
A project of support $V$ is a pair $(a,A)$ of a real number $a$ and a finite formal sum $A=\sum_{i\in I} \alpha_{i}A_{i}$ where for all $i\in I$, $\alpha_{i}\in\realN$ and $A_{i}$ is a graphing of support $V$.
\end{definition}

\begin{definition}
Two projects $(a,A)$ and $(b,B)$ are orthogonal -- written $(a,A)\poll{}(b,B)$ -- when they have equal support and $\meas[]{(a,A),(b,B)}\neq 0,\infty$. We also define the orthogonal of a set $E$ as $E^{\pol}=\{(b,B): \forall (a,A)\in A, (a,A)\poll{}(b,B)\}$ and write $E^{\pol\pol}$ the double-orthogonal $(E^{\pol})^{\pol}$.
\end{definition}

Based on this orthogonality relation, we can define the notion of \emph{conducts} and \emph{behaviours} which are the interpretations of types in the models. 

\begin{definition}
A \emph{conduct} of support $V^{A}$ is a set $\cond{A}$ of projects of support $V^{A}$ such that $\cond{A}=\cond{A}^{\pol\pol}$. A \emph{behaviour} is a conduct such that whenever $(a,A)$ belongs to $\cond{A}$ (resp. $\cond{A}^{\pol}$) and for all $\lambda\in\realN$, then $(a,A+\lambda\emptyset)$ belongs to $\cond{A}$ (resp. $\cond{A}^{\pol}$) as well. If both $\cond{A}$ and $\cond{A}^{\pol}$ are non-empty, we say $\cond{A}$ is \emph{proper}.
\end{definition}

Conducts provide a model of Multiplicative Linear Logic. The connectives $\otimes,\multimap$ are defined as follows: if $\cond{A}$ and $\cond{B}$ be conducts of disjoint supports $V^{A}, V^{B}$, i.e. $V^{A}\cap V^{B}$ is of null measure, then:
\[
\begin{array}{rcl}
\cond{A\otimes B}&=&\{\de{a\plug b}~|~\de{a}\in\cond{A},\de{b}\in\cond{B}\}^{\pol\pol}\\
\cond{A\multimap B}&=&\{\de{f}~|~\forall \de{a}\in\cond{A},\de{f\plug b}\in\cond{B}\}
\end{array}
\]
However, to define additive connectives, one has to restrict the model to behaviours. In this paper, we will deal almost exclusively with proper behaviours. Based on the following proposition, we will therefore consider mostly projects of the form $(0,L)$ which we abusively identify with the underlying sliced graphing $L$. Moreover, we will use the term \enquote{behaviour} in place of \enquote{proper behaviour}.

\begin{proposition}[{\cite[Proposition 60]{seiller-goiadd}}]
If $\cond{A}$ is a proper behaviour, $(a,A)\in\cond{A}$ implies $a=0$.
\end{proposition}

Finally, let us mention the fundamental theorem for the interaction graphs construction in the restricted case we just exposed\footnote{The general construction allows for other sets of weights as well as whole families of measurements \cite{seiller-goig}.}.

\begin{theorem}[{\cite[Theorem 1]{seiller-goig}}]\label{thm:mallmodels}
For any microcosm $\microcosm{m}$, the set of behaviours provides a model of Multiplicative-Additive Linear Logic (\MALL) without multiplicative units.
\end{theorem}

This theorem can be refined, as the set of conducts provides a model of Multiplicative Linear Logic (\MLL), although multiplicative units are \emph{not} behaviours. Moreover \MALL is only the minimal fragment one can expect to model, and one can define models which interpret second-order quantification \cite{seiller-goig} as well as exponential connectives \cite{seiller-goie,seiller-goif}.

\section{Integers, Machines, Tests}

We will now define a specific model that will be studied throughout the rest of the paper. After defining the underlying measure space, we will define a family of microcosms. The largest of those microcosms will be used to define the model with which we will work -- the surrounding universe. We will start by showing that this is a model Elementary Linear Logic (\ELL), a logic fragment expressive enough to define a representation of binary words. The smaller microcosms $\microcosm{m}_{i}$ will be used to define submodels of this surrounding universe which will characterise small complexity classes.

Since the resulting model is of Elementary Linear Logic (\ELL), one can represent binary words using the type of binary lists in \ELL. The corresponding proofs can then be interpreted as graphings (or rather as projects $(0,G)$ with $G$ a graphing), but a single proof can be interpreted as a myriad of graphings depending on choice in the interpretation's definition. Consequently, an \ELL proof representing a binary word will be interpretable by many different graphings. Those graphings, however, are all obtained as representations of the same graph, corresponding to the set of axiom rules in the corresponding proof net. We refer the reader to an earlier paper for an illustrated discussion of how binary words can be represented as graphs \cite{seiller-conl}; we define in the next section the interpretation of binary words directly. 

Once the type of binary words is defined, one can consider the type of binary predicates in the model. Among those graphings realising this type, we consider only the \emph{finite} ones, i.e. those that can be described by a finite number of edges. These objects are called \emph{machines}, and can be further classified according to the monoids of measurable maps used to realise their edges. This leads to a notion of $\microcosm{m}$-machine for a microcosm $\microcosm{m}$ which is a submonoid of $\microcosm{p}$. In a way, we are therefore defining subsets of the type of predicates in a model of \ELL. However, let us recall that each such submonoid $\microcosm{m}$ describes a model of \MALL (at least); consequently another reading of this is to understand $\microcosm{m}$-machines as finite graphings in the type of predicates of a smaller model described by $\microcosm{m}$. In particular, these models are not complete w.r.t. \MALL and should satisfy additional axioms. Since these models characterise small complexity classes, one could try to derive from these models logical systems describing (space) sub-linear complexity classes.

\subsection{General Situation}

Notice that while previous work (and the previous section) defined graphing with weights in an arbitrary monoid $\Omega$, we here fix $\Omega$ as $[0,1]\times\{0,1\}$ with usual multiplication on the unit interval and the product on $\{0,1\}$. To simplify notations, we write elements of the form $(a,0)$ as $a$ and elements of the form $(a,1)$ as $a\cdot \mathbf{1}$. On this set of weights, we will consider the fixed parameter map $m(x,y)=xy$ in the following.

In practice, most graphings considered in this paper do not use weights different from $1$ (i.e. $(1,0)$), except for the \emph{tests} (\autoref{def:tests}). We will therefore allow ourselves to define graphing representatives without mentioning the weights, implying that those are all equal to $1$.

Moreover, graphings were shown equivalent w.r.t. dialect-renaming, i.e. if $G$ is obtained from $F$ by renaming the dialect then $F$ and $G$ are \emph{universally equivalent} \cite{seiller-goie}, i.e. indistinguishable in the model. Formally, this is expressed as the fact that for every graphing $H$, the measurement $\meas{F,H}$ coincides with the measurement $\meas{G,H}$. Consequently, we will always consider in the following that dialects are chosen as initial segments of the natural numbers, i.e. sets $[n]=\{0,1,\dots,n\}$.

\begin{definition}[The Space]
We will work on the measure space $\measure{X}=\integerN\times[0,1]^{\naturalN}$ considered with its usual Borel $\sigma$-algebra and Lebesgue measure.
\end{definition}

Borrowing the notation introduced in earlier work \cite{seiller-goif}, we denote by $(x,\seq{s})$ the points in $\measure{X}$, where $\seq{s}$ is a sequence for which we allow a concatenation-based notation, i.e.\ we write $(a,b)\cdot \seq{s}$ the sequences whose first two elements are $a,b$ (and we abusively write $a\cdot\seq{s}$ instead of $(a)\cdot\seq{s}$). Given a permutation $\sigma$ over the natural numbers, we write $\sigma(\seq{s})$ the result of its natural action on the $\naturalN$-indexed list $\seq{s}$.

\begin{definition}[Microcosms]
For all integer $i\geqslant 1$, we consider the microcosm $\mathfrak{m}_{i}$ generated by the translations $\ttrm{t}_{z}:(x,\seq{s})\mapsto(x+z,\seq{s})$ for all integer $z$, and the permutations $\ttrm{p}_{\sigma}:(x,\seq{s})\mapsto(x,\sigma(\seq{s}))$ for all permutation $\sigma$ such that $\sigma(k)=k$ for all $k>i$. We write $\microcosm{m}_{\infty}$ the union $\cup_{i>1}\microcosm{m}_{i}$.

We also define the microcosms $\bar{\microcosm{m}}_{i}$ as the smallest microcosm containing $\microcosm{m}_{i}$ and all translations\footnote{We denote here by $a\bar{+}b$ the fractional part of the sum $a+b$.} $\ttrm{t}_{\lambda}:(x,a\cdot \seq{s})\mapsto(x, (a\bar{+}\lambda)\cdot\seq{s})$ for $\lambda$ in $[0,1]$.
\end{definition}

We now define a bijective measure-preserving pairing function: $[\cdot,\cdot]:[0,1]^{2}\rightarrow[0,1]$. Although it will not be used in the next sections, this will help us draw the connection between the present results and models of Elementary Linear Logic.

\noindent Given a subset $A$ of $\measure{X}$, integers $d<n$, we define the set:
$$\pushed{d}(A)=\{(a,[x,y]\cdot\seq{s}): (a,\seq{s})\in A, d\leqslant nx\leqslant d+1, y\in[0,1]\}.$$
Given a measurable map $f:A\rightarrow B$ and integers $d,d'<n$, we define the measurable map:
$$\pushed{d,d'}(f):\left\{\begin{array}{rcll}
		\pushed{d}(A)&\rightarrow&\pushed{d'}(B)\\
		(a,x\cdot\seq{s})&\mapsto&(a',y\cdot\seq{s'})&\text{($(a',\seq{s'})=f(a,\seq{s})$, $y=x+(d'-d)/n$)}
		\end{array}\right.
		$$

\begin{definition}
Given a graphing $G=\{(S^{G}_{e},\ttrm{i}^{G}_{e},\ttrm{o}^{G}_{e},\phi^{G}_{e})$ of dialect $D=[n]$, we define the promotion $\oc G$ of $G$ as the following graphing of dialect $[0]$:
$$\{(\pushed{\ttrm{i}^{G}_{e}}(S^{G}_{e}),0,0,\pushed{\ttrm{i}^{G}_{e},\ttrm{o}^{G}_{e}}(\phi^{G}_{e}))~|~e\in E^{G}\}$$
\end{definition}

This previous definition is a \emph{perennisation}, as defined in earlier papers \cite{seiller-phd,seiller-goie}, i.e. it maps arbitrary graphings to graphings with trivial dialect $[0]$. This is to ensure that all graphings of the form $\oc A$ are \emph{duplicable}: since one can always find a graphing $C$ such that $C\plug A\simeq A\otimes A$ for all $A$ with a trivial dialect \cite[Proposition 36]{seiller-goie}, we can implement \emph{contraction} on graphings of the form $\oc A$, and by extension on conducts generated by graphings of this form.

\begin{definition}
Given a behaviour $\cond{A}$, we define the conduct $\cond{\oc A}$ as the set $\{(0,\oc G)~|~ G\in\cond{A}\}^{\pol\pol}$.
\end{definition}

Following the remark above, given any conduct $\cond{A}$ one can always define a graphing $C$ implementing contraction, i.e. such that $(0,C)\in\cond{\oc A\multimap \oc A\otimes \oc A}$.

\begin{remark}
It is important to note that the conduct $\cond{\oc A}$ \emph{never} is a behaviour. However, if $\cond{B}$ is an arbitrary behaviour, $\cond{\oc A\multimap B}$ is a behaviour \cite[Corollary 57]{seiller-goie}.
\end{remark}

\begin{theorem}\label{theoremELL}
Consider the microcosm $\microcosm{p}$ generated by $\bar{\microcosm{m}}_{\infty}$ together with the additional maps $\pair$ and $\pair^{-1}$, where $\pair: (a,(x,y)\cdot\seq{s})\mapsto (a,[x,y]\cdot\seq{s})$. For any microcosm containing $\microcosm{p}$, the set of conducts and behaviours is a model of Elementary Linear Logic.
\end{theorem}

\begin{proof}
We only need to check that functorial promotion can be implemented, as contraction is automatically satisfied \cite{seiller-goie} and the fact that it is a model of \MALL follows from \autoref{thm:mallmodels}. The technique is similar as what is used in previous papers \cite{seiller-goie,seiller-goif}. First, we notice the maps $\lin=(a,[x,[y,z]]\cdot\seq{s})\mapsto(a,[[x,y],z]\cdot\seq{s})$ and $\rin=(a,[x,[y,z]]\cdot\seq{s})\mapsto(a,[[y,x],z]\cdot\seq{s})$ % as $\pair\circ\pair\circ\sigma_{(1,2)}\circ\sigma_{(2,3)}\circ\pair^{-1}\circ\sigma_{(1,2)}\circ\pair^{-1}$ and $\pair\circ\pair\circ\sigma_{(2,3)}\circ\pair^{-1}\circ\sigma_{(1,2)}\circ\pair^{-1}$ respectively.
 belong to the microcosm $\microcosm{p}$. Then, given $F\in\cond{A\multimap B}$ and $A\in\cond{A}$, we can check that $\lin(\oc F)\plug\rin(\oc A)$ is equivalent to $\oc (F\plug A)$, which is an element of $\cond{\oc B}$.
\end{proof}

%We now have introduced the necessary material to define the representation of binary words. 

\subsection{Representation of Binary Words}

We use here the \ELL encoding of binary words, i.e. as elements of the type $\ListTypeELL=\forall X, \oc (X\multimap X)\multimap \oc (X\multimap X)\multimap \oc (X\multimap X)$. We write $\Sigma=\{0,1\}$, and denote by $\starred{\Sigma}$ the extended alphabet $\Sigma\cup\{\star\}$: a binary word $\word{w}$ will be represented with a starting symbol $\star$, i.e. $\word{w}=\star a_{1} a_{2} \dots a_{n}$ where $a_{i}\in\Sigma$.

\begin{notation}
We write $\ext{\Sigma}$ the set $\starred{\Sigma}\times\{\In,\Out\}$. We also denote by $\vertices{\Sigma}$ the set $\ext{\Sigma}\cup\{\ttrm{a},\ttrm{r}\}$, where $\ttrm{a}$ (resp. $\ttrm{r}$) stand for $\accept$ (resp. $\reject$). 
\end{notation}

\begin{notation}
We fix once and for all an injection $\Psi$ from the set $\vertices{\Sigma}$ to intervals in $\mathbf{R}$ of the form $[k,k+1]$ with $k$ an integer. For all $f\in\vertices{\Sigma}$ and $Y$ a measurable subset of $[0,1]^{\naturalN}$, we denote by $\bracket{f}_{Y}$ the measurable subset $\Psi(f)\times Y$ of $\measure{X}$. If $Y=[0,1]^{\naturalN}$, we omit the subscript and write $\support{f}$. The notation extends to any subset $S$ of $\vertices{\Sigma}$, i.e.\ $\support{S}$ is the (disjoint) union $\cup_{f\in S}\support{f}$.
\end{notation}

Given a word $\word{w}=\star a_{1}a_{2}\dots a_{k}$, we denote $\graphtw{w}$ the graph with set of vertices $V^{\graphtw{w}}\times D^{\graphtw{w}}$, set of edges $E^{\graphtw{w}}$, source map $s^{\graphtw{w}}$ and target map $t^{\graphtw{w}}$ respectively defined as follows:
\begin{equation*}
\begin{array}{cc}
\begin{array}{rcl}
V^{\graphtw{w}}&=&\ext{\Sigma}\\
D^{\graphtw{w}}&=&[k]\\
E^{\graphtw{w}}&=&\{r,l\}\times[k]\\
\end{array}
&
\begin{array}{rcl}
s^{\graphtw{w}}&=&(r,i)\mapsto (a_{i},\Out,i)\\
		&&(l,i)\mapsto (a_{i}, \In,i)\\
t^{\graphtw{w}}&=&(r,i)\mapsto (a_{i+1},\In,i+1 \textnormal{ mod } k+1)\\
		&&(l,i)\mapsto (a_{i-1}, \Out,i-1 \textnormal{ mod } k+1)
\end{array}
\end{array}
\end{equation*}
This graph is the discrete representation of $\word{w}$. Detailed explanations on how these graphs relate to the proofs of the formula $\ListTypeELL$ can be found in earlier work \cite{seiller-phd,seiller-conl}.

\begin{definition}
Let $\word{w}$ be a word $\word{w}=\star a_{1}a_{2}\dots a_{k}$ over the alphabet $\Sigma$. We define the word graphing $\graphingtw{w}$ of support $\wordsupport$ and dialect $D^{\graphtw{w}}$ by the set of edges $E^{\graphtw{w}}$ and for all edge $e$:
$$\{(\bracket{f},i,j,\phi_{f,i}^{g,j},1): e\in E^{\graphtw{w}}, s^{\graphtw{w}}(e)=(f,i), t^{\graphtw{w}}(e)=(g,j), \phi_{f,i}^{g,j}:(\bracket{f},x,i)\mapsto (\bracket{g},x,j)\}$$
%The graphing $\graphingow{w}$ is easily defined in a similar way as a realiser of $\graphow{w}$.
\end{definition}

\begin{notation}
We write $\reptw{w}$ the set of word graphings for $\word{w}$. It is defined as the %the set of all \emph{variants} of the graphing $\graphingtw{w}$, i.e. the 
set of graphings obtained by renaming the dialect $D^{\graphtw{w}}$ w.r.t. an injection $[k]\rightarrow [n]$.% We will use the notation $\repany{w}$ when talking about any of these sets of representations.
\end{notation}

\begin{definition}
Given a word $\word{w}$, a \emph{representation of $\word{w}$} is a graphing $\oc L$ where $L$ belongs to $\reptw{w}$. The set of representations of words in $\Sigma$ is denoted $\twwprojects$, the set of representations of a specific word $\word{w}$ is denoted $\repany{w}$. 

We then define the conduct $\oc\ListType=\cond{(\twwprojects)^{\pol{}\pol{}}}$.
%where $(\twwprojects)^{+}$ is the set of projects $\de{a}+\lambda_1\de{0}+\dots+\lambda_n\de{0}$ with $\de{a}\in\twwprojects$, $n\in\naturalN$ and $\lambda_1,\dots\lambda_n\in\realN$.
\end{definition}

\begin{definition}
We define the (unproper) behaviour $\NBool$ as $\cond{T}_{\resultsupport}$, where for all measurable set $V$ the behaviour $\cond{T}_{V}$ is defined as the set of all projects of support $V$.
For all microcosm $\microcosm{m}$, we define $\pred{m}$ as the set of $\microcosm{m}$-graphings in $\oc\ListType\multimap\NBool$.
\end{definition}

%In the following we will consider sets of elements of the behaviour $\pred{p}$ in this model which are realised in a submicrocosm $\microcosm{m}_{i}$ of $\microcosm{p}$. These sets of elements can thus be considered either as a subset of the behaviour $\pred{p}$ in the \ELL model or as a behaviour in itself since it is equal to the behaviour $\oc\ListType\multimap\NBool$ in the model described by $\microcosm{m}_{i}$.

\subsection{Predicate Machines and Tests}

We now turn to the notion of machine. We focus in this paper on machines computing predicates, i.e. elements of the type $\oc\ListType\multimap\NBool$. Computing devices are traditionally discrete and finite objects, and it is therefore quite natural to envision them as graphs. However, the notion we consider -- called $\microcosm{m}$-machines -- will be \emph{realisations} of graphs as $\microcosm{m}$-graphings, i.e.\ infinite objects in some ways. Intuitively, the underlying graph corresponds to the simple notion of automaton (with the dialect playing the role of control states), while the realisations of edges correspond to particular \emph{instructions}. This intuitive understanding of $\microcosm{m}$-machines can be followed through the rest of this paper.

\begin{definition}
A graphing $G$ is \emph{finite} when there exists a graphing $H$ such that\footnote{We use the notation $G\leqslant H$ for \enquote{F is a refinement of $G$} for the notion of refinement explained in \autoref{subsec:basics}.} $G\leqslant H$ and the set of edges $E^{H}$ is finite.
\end{definition}

\begin{definition}
%\note{A $(\Omega,\microcosm{m})$-machine over the alphabet $\Sigma$ is a $\microcosm{m}$-realisation of a finite $\Omega$-weighted graph over the set of vertices $\machinesupport$.}
A \emph{nondeterministic predicate $\microcosm{m}$-machine} over the alphabet $\Sigma$ is a finite $\microcosm{m}$-graphing belonging to $\pred{m}$ with all weights equal to $1$.
\end{definition}

The computation of a given machine given an argument is represented by the \emph{execution}, i.e. the computation of paths defined in \autoref{subsec:execution}. The result of the execution is an element of $\NBool$, i.e.\ in some ways a generalised boolean value\footnote{If one were working with \enquote{deterministic machines} \cite{seiller-towards}, it would belong to the subtype $\Bool$ of booleans.}.

\begin{definition}[Computation]
Let $M$ be a $\microcosm{m}$-machine, $\word{w}\in\Sigma^{\ast}$ and $\oc L\in\oc\ListType$. The \emph{computation} of $\de{M}$ over $\oc L$ is defined as the graphing $M\plug \oc L$, an element of $\NBool$.
\end{definition}

We now introduce the notion of \emph{test}. This notion is essential as it allows for considering several notions of acceptance. Even though acceptance may be defined \enquote{by hand} by describing directly the expected result, the definition through tests allows for a more interesting definition. Indeed, the acceptance is described inside the model, using already existing notions, i.e. we do not modify the models to define testing. In other words, acceptance and rejection are given a logical meaning, as testing is tied with the process of constructing types.

\begin{definition}[Tests]\label{def:tests}
A \emph{test} is a family $\testfont{T}=\{\de{t}_{i}=(t_{i},T_{i})~|~i\in I\}$ of projects of support $\resultsupport$.% satisfying the following invariance condition: for all $(t_{i},T_{i})\in\testfont{T}$ and pairs $B_{\ttrm{a}}$, $B_{\ttrm{r}}$ of Borel automorphisms of $\support{a}$, $\support{r}$ respectively, the project $(t_{i},(B_{\ttrm{a}}\cup B_{\ttrm{r}}(T_{i}))$ lies in $\testfont{T}$. %A $(\Omega,\microcosm{m})$-machine \emph{accepts} the word $\word{w}$ with respect to the test $\testfont{T}$ if and only if the computation $M\plug L_{l}$ is orthogonal to any element of $\testfont{T}$.
%If $\alpha$ is an action of the monoid $\Omega$ on the space $X$, we define the associated graphing $[T]_{\alpha}=\{\phi^{\omega}_{e}:S^{\omega}_{e}\rightarrow T^{\omega}_{e}\}_{e\in E^{T}}$, where $S_{e}^{\omega}=(s^{T}(e),X,\omega)$, $T_{e}^{\omega}=(t^{T}(e),X,\omega,\tt s\rm)$ and $\phi^{\omega}_{e}$ is the identity on $X$ and $\omega$.
\end{definition}

We now want to define the language characterised by a machine. For this, one could consider \emph{existential} $\mathcal{L}_{\exists}^{\testfont{T}}(M)$ and  \emph{universal} $\mathcal{L}_{\forall}^{\testfont{T}}(M)$ languages for a machine $M$ w.r.t. a test $\testfont{T}$:
$$
\begin{array}{rcl}
\mathcal{L}_{\exists}^{\testfont{T}}(M)&=&\{\word{w}\in\Sigma^{\ast}~|~\forall \de{t}_{i}\in\testfont{T}, \exists \de{w}\in\repany{w}, M\plug \de{w}\poll{} \de{t}_{i}\}\\
\mathcal{L}_{\forall}^{\testfont{T}}(M)&=&\{\word{w}\in\Sigma^{\ast}~|~\forall \de{t}_{i}\in\testfont{T}, \forall \de{w}\in\repany{w}, M\plug \de{w}\poll{} \de{t}_{i}\}
\end{array}
$$

We now introduce the notion of uniformity, which describes a situation where both definitions above coincide. This collapse of definitions is of particular interest because it ensures that both of the following problems are easy to solve:
\begin{itemize}[nolistsep,noitemsep]
\item whether a word belongs to the language: from the existential definition one only needs to consider one representation of the word;
\item whether a word does not belong to the language: from the universal definition, one needs to consider only one representation of the word.
\end{itemize}

\begin{definition}[Uniformity]
Let $\microcosm{m}$ be a microcosm. The test $\testfont{T}$ is said \emph{uniform} w.r.t. $\microcosm{m}$-machines if for all such machine $M$, and any two elements $\de{w},\de{w'}$ in $\repany{w}$:
$$M\plug \de{w}\in \testfont{T}^{\pol} \text{ if and only if }M\plug \de{w}\in \testfont{T}^{\pol}$$
Given a $\microcosm{m}$-machine $M$, we write in this case $\mathcal{L}^{\testfont{T}}(M)=\mathcal{L}_{\exists}^{\testfont{T}}(M)=\mathcal{L}_{\forall}^{\testfont{T}}(M)$.
\end{definition}

\section{Characterising a nondeterministic Hierarchy}

\subsection{Multihead Automata}

We consider a variant of the classical notion of two-way multihead finite automata obtained by:
\begin{itemize}[nolistsep,noitemsep]
\item fixing the right and left end-markers as both being equal to the fixed symbol $\star$;
\item fixing once and for all unique initial, accept and reject states;%, respectively denoted by $\init$, $\accept$ and $\reject$;
\item choosing that each transition step moves exactly one of the multiple heads of the automaton;
\item imposing that all heads are repositioned on the left end-marker before accepting/rejecting.
\end{itemize}
It should be clear that these choices in design have no effect on the sets of languages recognised. %Indeed, the first variation -- collapsing the two end-markers into a single symbol -- can be dealt with by doubling the set of states in order to keep the information of which end-marker one is likely to be encountered. I.e. if I'm moving from left-to-right I can only encounter the right-hand endmarker. 

\begin{definition}
A \emph{two-way multihead automaton} $\automaton{M}$ with $k$ heads is defined as a tuple $(\Sigma,Q,\rightarrow)$, where $\rightarrow\subseteq \left(\starred{\Sigma}^{k}\times Q\right)\times\left((\{1,\dots,k\}\times\{\In,\Out\})\times Q\right)$ is the \emph{transition relation} of $\automaton{M}$. The automaton $\automaton{M}$ is \emph{deterministic} when the relation $\rightarrow$ is functional.

The set of two-way multihead automata with $k$ heads is written $\twnfa{k}$, and the set of all two-way multihead automata $\cup_{k\geqslant 1}\twnfa{k}$ is denoted by $\twnfas$. %The set of deterministic two-way multihead automata with $k$ heads is written $\twdfa{k}$, and the set of all deterministic two-way multihead automata $\cup_{k\geqslant 1}\twdfa{k}$ is denoted by $\twdfas$.
\end{definition}

\begin{definition}
We denote \cctwconfa{k} the set of languages accepted by automata in $\twnfa{k}$, where an automaton $\automaton{M}$ accepts a word $\word{w}$ if and only there are no computation trace of $\automaton{M}$ given $\word{w}$ as input leading to a rejecting state.
\end{definition}

The set of languages \Regular= \cctwconfa{1} is usually called the set of \emph{regular languages}. We now state two of the main results in the theory of two-way multihead automata.

\begin{theorem}[Monien \cite{monien}]\label{thm:monien}
For all $k$, the set \cctwconfa{k} is a \emph{strict} subset of \cctwconfa{k+1}. 
\end{theorem}

\begin{theorem}
$\cup_{k\geqslant 1} \textnormal{\cctwconfa{k}} = \textnormal{\coNLogspace}.$
\end{theorem}

We will now show how $k$-head multihead automata corresponds to $\microcosm{m}_{k}$-machines. The reader will find some examples of graphing representations of integers, machines, and computations in an overview and perspective paper by the author \cite{seiller-towards}.

\subsection{Automata as Machines}

There are two main differences between the model of multihead automata with $k$ heads and the notion of $\microcosm{m}_{k}$-machines.
\begin{itemize}[noitemsep,nolistsep]
%\item The first difference is that we now have a circular input with a single endmarker instead of a non-circular input with two endmarkers. This difference is not a significative one; it is easy to deal with this issue by considering an exetended set of states. Therefore, in order to deal with this first distinction, an automaton $\automaton{M}$ with a set of states $Q$ will be realised by a $(\Omega,\microcosm{m}_{i})$-machine with $Q\times\{\In,\Out\}^{i}$ states. 
\item The first difference is that when one \enquote{moves the $i$-th head} of a $\microcosm{m}_{k}$-machine, it induces a reindexing of the sets of heads. I.e. a $\microcosm{m}_{k}$-machine should be understood as a multihead automata that can only move its principal head, but has the possibility of reindexing its heads following any permutation over $k$ elements. To deal with this, %we could consider (yet another) variation on multihead automata to make them mimic $(\Omega,\microcosm{m}_{k})$-machines and then show they are equivalent to the usual notion of multihead automata (which they are). We chose however to interpret directly the usual notion of multihead automata as $(\Omega,\microcosm{m}_{k})$-machines, and conversely, dealing with this issue in the proofs. We 
we will extend the set of states $Q$ of the automaton we wish to represent and consider $\bar{Q}=Q\times \mathfrak{G}_{k}$; the set of permutations $\mathfrak{G}_{k}$ being used to keep track of the heads' reindexings.
\item The second difference comes from the fact that the computation of $\microcosm{m}_{k}$-machines is \enquote{dynamic}, i.e. corresponds to a dialogue between the machine and the representation of the word it is given as input. As a consequence, one has the knowledge of what symbol a given pointer is reading at a given location \emph{only at the exact moment the pointer moves onto this location}. I.e. the pointer receives information about the input from the integer, and one has to store it if it is to be reused later on. This is different from the way multihead automata compute since the latter can, at any given time, access the value located where any head is pointing at. To take care of this difference, we extend once again the set of states. As a consequence, the automaton $\automaton{M}$ with a set of states $Q$ will be realised as a $\microcosm{m}_{k}$-machine with an extended set of states (encoded as the \emph{dialect}) $\bar{Q}\times\{\star,0,1\}^{k}$. 
\end{itemize}

\begin{definition}
Let $\automaton{M}$ be an automaton with $k$ heads. We here write $\automaton{M}=(\Sigma,Q,\rightarrow)$. We define $\autograph{M}$ a graphing in $\microcosm{m}_{k}$ with dialect -- set of states -- $Q\times\mathfrak{G}_{k}\times\{\star,0,1\}^{k}$ as follows. 

\noindent The set of edges of $\autograph{M}$ is the set:
\[E^{\autograph{M}}=\{(\autograph{t},a,d,\sigma)~|~ t\in\rightarrow, a\in\{\star,0,1\},d\in\{\textnormal{in},\textnormal{out}\}, \sigma\in\mathfrak{G}_{k}\}\]
%\note{to which we add a number of additional edges corresponding to the choices of the initial, accepting and rejecting states of the automata.} 

\noindent The source of the edge $(\autograph{t},s,d,\sigma)$ for $t=((\vec{s},q),(i,d',q'))$ is defined as:
\[
S_{(\autograph{t},a,d,\sigma)}=
\left\{\begin{array}{ll}
\support{(s,d)}\times\{(q,\sigma,\vec{s})\}&
\text{ if }q\neq \textnormal{init}\\
\support{\ttrm{a}}\times\{(\textnormal{init},\identity,\vec{\star})\}&
\text{ if }q=\textnormal{init}\text{ and }d=\textnormal{in}\\
\support{\ttrm{r}}\times\{(\textnormal{init},\identity,\vec{\star})\}&
\text{ if }q=\textnormal{init}\text{ and }d=\textnormal{out}
\end{array}\right.
\]
%\begin{itemize}[noitemsep,nolistsep]
%\item the set $\support{(s,d)}\times\{(q,\sigma,\vec{s})\}$ when $q\neq \textnormal{init}$; 
%\item the set $\support{\ttrm{a}}\times\{(\textnormal{init},\identity,\vec{\star})\}$ if $q=\textnormal{init}$ and $d=\textnormal{in}$;
%\item the set $\support{\ttrm{r}}\times\{(\textnormal{init},\identity,\vec{\star})\}$ if $q=\textnormal{init}$ and $d=\textnormal{out}$.
%\end{itemize}
 
\noindent The target of the edge $(\autograph{t},s,d,\sigma)$ for $t=((\vec{s},q),(i,d',q'))$ is defined as:
\[
T_{(\autograph{t},a,d,\sigma)}=
\left\{\begin{array}{ll}
\support{(s_{i},d')}\times\{(q',\tau_{1,\sigma(i)}\circ\sigma,\vec{s}[s_{\sigma^{-1}(1)}:=s])\}&
\text{ if }q'\not\in\{\textnormal{accept,reject}\}\\
\support{q'}\times\{(\textnormal{init},\identity,\vec{\star})\}&
\text{ if }q'\in\{\textnormal{accept,reject}\}
\end{array}\right.
\]
%\begin{itemize}[noitemsep,nolistsep]
%\item the set $\support{(s_{i},d')}\times\{(q',\tau_{1,\sigma(i)}\circ\sigma,\vec{s}[s_{\sigma^{-1}(1)}:=a])\}$ when $q'\not\in\{\textnormal{accept,reject}\}$;
%\item the set $\support{q'}\times\{(\textnormal{init},\identity,\vec{\star})\}$ if $q'\in\{\textnormal{accept,reject}\}$.
%\end{itemize}

\noindent The realiser of the edge $(\autograph{t},a,d,\sigma)$ for $t=((\vec{s},q),(i,d',q'))$ is the map $\ttrm{p}_{(1,\sigma(i))}$ composed with the adequate translation on $\integerN$. E.g. when $q\neq \textnormal{init}$ and $q'\not\in\{\textnormal{accept,reject}\}$ it is the map $\ttrm{p}_{(1,\sigma(i))}$ composed with the bijection exchanging $\support{(s,d)}$ and $\support{(s_{i},d')}$. %The additional edges are realised by simple translations between $\support{\accept}$ and the initial state, between accepting states and $\support{\accept}$ and between rejecting states and $\support{\reject}$.
\end{definition}

Let us explain how this encoding simulates the automaton. We fix a word $\word{w}=\star a_{1}a_{2}\dots a_{n}$ and a configuration $\ttrm{C}$ of a $k$-head automaton, i.e. a sequence of heads positions $(p_{i})_{i=1}^{k}$ -- where for all $i$, $p_{i}\in\{0,\dots,n\}$ --, and a state $q$. Depending on the value $\vec{s}=a_{p_{1}},\dots,a_{p_{k}}$, the automaton will fire different transitions. Let us pick one, namely $\ttrm{t}=(\vec{s},q)\rightarrow(i,d',q')$. There is a family of corresponding edges in the automaton, denoted by $(\{\ttrm{t}\},a,d,\sigma)$. Here, $\sigma$ is a permutation that remembers how heads have been reindexed since the initial transition; as explained above, this is because moving a head requires a reindexing. The pair $(s,d)$ records a symbol and a direction, namely the symbol and direction of the previous transition made by the automaton: it is therefore uniquely fixed when considering a given computation trace. Then a given edge $(\{\ttrm{t}\},a,d,\sigma)$ maps the set $\support{(s,d)}\times\{(q,\sigma,\vec{s})\}$ to $\support{(s_{i},d')}\times\{(q',\tau_{1,\sigma(i)}\circ\sigma,\vec{s}[s_{\sigma^{-1}(1)}:=a])\}$ (supposing $q\neq\init$ and $q'\neq\mathrm{accept}, \mathrm{reject}$). In doing so, it is updating the value of the sequence $\vec{s}$ according the value read by the pointer moved \emph{during the previous transition which lead to $(s,d)$}. It is also positioning its $i$th head adequately by reindexing it using the map $\ttrm{p}_{(1,\sigma(i))}$ and waiting for the integer to provide its next value in direction $d'$ by fixing the target subset $\support{(s_{i},d')}$ ($s_{i}$ being the last value read by the $i$-th head).

The following proposition is then proved by induction.% the definition of the graphing $\autograph{M}$ and is proved by an easy induction.

\begin{proposition}\label{tracespaths}
Let $\automaton{M}$ be a $k$-heads automaton. Alternating paths of odd length between $\autograph{M}$ and $\oc \graphingtw{w}$ of source $\support{\ttrm{a}}_{Y}$ (resp. of source $\support{\ttrm{r}}_{Y}$) with\footnote{To understand where the subset $Y$ comes from, we refer the reader to the proof of \autoref{technicallemma}.} $Y=[0,\frac{1}{\lg(\word{w})}]^{k}\times[0,1]^{\naturalN}$ are in bijective correspondence with the non-empty computation traces of $\automaton{M}$ given $\word{w}$ as input.
\end{proposition}

\begin{corollary}\label{pathaccept}
The automaton $\automaton{M}$ accept the word $\word{w}$ if and only if there exists no alternating path between $\autograph{M}$ and $\oc \graphingtw{w}$ from $\support{\ttrm{r}}$ to itself.
\end{corollary}

\begin{definition}
We define the test $\testdetneg$ as the set consisting of the projects $\de{t}^{-}_{\zeta}=(\zeta,\identity[\support{\ttrm{r}}])$, where $\zeta\neq0$ and $\identity[\support{\ttrm{r}}]$ is the graphing with a single edge and trivial dialect $[0]$: $\{(\support{\ttrm{r}},0,0,x\mapsto x,1\cdot\mathbf{1})\}$.
\end{definition}

The fact that this test is uniform comes from the invariance of the underlying graphing $\identity[\support{\ttrm{r}}]$ w.r.t. any bijective transformation. In more details, two representations of the same integer $\oc W$ and $\oc W'$ can be shown to relate through a measurable (though not measure-preserving) bijection $\theta$ by conjugation, i.e. $\phi\mapsto \theta^{-1}\phi\theta$ maps edges in $\oc W$ to edges in $\oc W'$. Then, one just has to remark that the realiser of an alternating path between $\oc W'$ and $\identity[\support{\ttrm{r}}]$ contains subsequences of the form $\theta\circ\theta^{-1}$ which shows, by simplification, that there exists a corresponding path alternating between $\oc W$ and $\identity[\support{\ttrm{r}}]$.

\begin{proposition}
The test $\testdetneg$ is uniform w.r.t. $\microcosm{m}_{\infty}$-machines.
\end{proposition}

%\begin{proposition}\label{testpathsneg}
%A project $(0,R)$ of support $\support{\ttrm{a},\ttrm{r}}$ is orthogonal to $\testdetneg$ if and only if there are no cycles going through $\support{\ttrm{r}}$ in $R$.
%\end{proposition}
%
%\begin{proof}
%The key point here is that $\meas[]{A,B}\neq0$ as soon as there is a cycle between $A$ and $B$. Suppose that $(0,R)$ is orthogonal to $\testdetneg$. Then for all $\zeta\neq0$, $\zeta+\meas[]{R,\identity[\support{\ttrm{r}}]}\neq 0,\infty$. If there is a cycle going through $\support{\ttrm{r}}$ in $R$, then it creates a cycle with $\identity[\support{\ttrm{r}}]$, and the measurement $\meas[]{R,\identity[\support{\ttrm{r}}]}$ is different from $0$, leading to a contradiction as it is either $\infty$ or can be canceled with the an adequate choice of $\zeta$.
%
%Conversely, if there are no edges from $\support{\ttrm{r}}$ to itself in $R$, then the measurement $\meas[]{R,\identity[\support{\ttrm{r}}]}$ is equal to $0$. Hence for all $\zeta\neq 0$, $\zeta+\meas[]{R,\identity[\support{\ttrm{r}}]}=\zeta\neq 0,\infty$, i.e. $(0,R)$ is orthogonal to $\testdetneg$.
%\end{proof}

\begin{proposition}
Let $\automaton{M}$ be a $\twnfa{k}$, $\word{w}$ a word. Then $\word{w}\in\mathcal{L}^{\testdetneg}(\autograph{M})$ if and only if $\automaton{M}$ accepts $\word{w}$.
\end{proposition}

\begin{proof}
From \autoref{tracespaths} and the constraint on automata that they should reinitialise their pointer to the left end-marker before accepting or rejecting, we know that $R=\autograph{M}\plug \oc \graphingtw{w}$ contains exactly as many edges from $\support{r}_{Y}$ to $\support{r}_{Y}$ -- here $Y$ is defined as in the statement of \autoref{tracespaths} -- as there are rejecting computation traces of $\automaton{M}$ given $\word{w}$ as input. 

Moreover, $\meas{\autograph{M},\oc \graphingtw{w}}$ is equal to $0$ as all weights of these graphings are equal to $1$. Then the result of the computation $(0,R)$ is orthogonal to $\testdetneg$ if and only if $\xi+\meas{R,\{\identity{\support{\ttrm{r}}}\}}\neq 0,\infty$ for all $\xi\neq0$. Now, this is true if and only if that $\meas{R,\identity{\support{\ttrm{r}}}\}}=0$, i.e. if and only if there are no edges from $\support{r}_{Y}$ to $\support{r}_{Y}$ in $R$ since any such edge creates a cycle with $\identity{\support{\ttrm{r}}}\}$ of weight $1\cdot\mathbf{1}$. 
\end{proof}

\begin{theorem}\label{soundness}
Any language computed by an acyclic $k$-head automaton is computed by a $\microcosm{m}_{k}$-machine w.r.t. $\testdetneg$.
\end{theorem}

\subsection{Machines as Automata}

We will now describe how one can define a $i$-head automaton computing the same language as any $\microcosm{m}_{i}$-machine. For this purpose, we will first restrict our attention to \emph{essential graphings}; i.e. graphings whose edges are realised by specific maps that correspond to a single instruction. Although the translation could be defined on general $\microcosm{m}_{i}$-machines, this restriction will help ease the formalisation.

\begin{definition}
A $\microcosm{m}$-machine $M$ is \emph{$\Gamma$-essential} w.r.t. a generating set $\Gamma$ of the microcosm $\microcosm{m}$ if every edge $e\in E^{M}$ is realised by a restriction of a map in $\Gamma$.
\end{definition}

\begin{theorem}\label{essentiallemma}
Let $\Gamma$ be a set of measurable maps, $\microcosm{m}$ the microcosm generated by $\Gamma$, and $M$ a $\microcosm{m}$-machine. There exists a $\Gamma$-essential $\microcosm{m}$-machine $\bar{M}$ such that, for all test $\testfont{T}$, $\mathcal{L}^{\testfont{T}}(M)=\mathcal{L}^{\testfont{T}}(\bar{M})$.
\end{theorem}

\begin{proof}
The proof is technical but not difficult. The principle is the following: one considers an extended dialect and then decomposes each edge that is not realised by an element of $\Gamma$ by a series of edges using specific new states (i.e. newly added elements of the dialect) and going back and forth on the input with the currently active head to stall the computation.
\end{proof}

The following is a technical lemma that uses some particular properties of the microcosm $\microcosm{m}_{\infty}$. This lemma is the equivalent, on our framework, to the so-called technical lemma which was essential in previous work involving operator algebras \cite{seiller-conl,seiller-lsp}.

\begin{lemma}[Technical Lemma]\label{technicallemma}
Let $M$ be a $\microcosm{m}_{\infty}$-machine. The computation of $M$ with a representation $\oc W$ of a word $\word{w}$ is the realisation by translations of a $\Omega$-weighted finite graph.
\end{lemma}

\begin{proof}
The proof of this lemma is based on the finiteness of $\microcosm{m}_{\infty}$-machines. Since $M$ is a finite graphing, there exits an integer $N$ such that $M$ is a $\microcosm{m}_{N}$-machine. We are thus left to prove the result for $M$ a $\microcosm{m}_{N}$-machine. We now pick a word $\word{w}\in\Sigma^{\ast}$, write $k$ the length of $\word{w}$ and $(0,W_{\word{w}})$ the project $(0,\oc \graphingtw{w})$. Let us remark that all maps realising edges in $M$ or in $\oc \graphingtw{w}$ are of the form $\phi\times\identity[\bigtimes_{i=N+1}^{\infty}{[0,1]}]$. We can therefore consider that the underlying space is $\integerN\times[0,1]^{N}$ instead of $\measure{X}$ by just replacing realisers $\phi\times\identity[\bigtimes_{i=N+1}^{\infty}{[0,1]}]$ by $\phi$. Moreover, the maps $\phi$ here act either as permutations over copies of $[0,1]$ (realisers of edges of $M$) or as permutations over a decomposition of $[0,1]$ into $k$ intervals (realisers of $\oc \graphingtw{w}$). Consequently, all realisers act as permutations over the set of $N$-cubes $\{\bigtimes_{i=1}^{N}[k_i/k,(k_{i}+1)/k]~|~0\leqslant k_i\leqslant k-1\}$, i.e. their restrictions to $N$-cubes are translations. 

Consequently, one can build two (thick\footnote{Thick graphs are graphs with dialects, where dialects act as they do in graphings, i.e.\ as control states.}) graphs $\bar{M}$ and $\bar{W}_{\word{w}}$ over the set of vertices $\ext{\Sigma}\times\{\bigtimes_{i=1}^{N}[k_i/k,(k_{i}+1)/k]~|~0\leqslant k_i\leqslant k-1\}$ as follows. There is an edge in $\bar{M}$ of source $(s,(k_{i})_{i=1}^{N},d)$ to $(s',(k'_{i})_{i=1}^{N},d')$ if and only if there is an edge in $M$ of source $\bracket{s}\times\{d\}$ and target $\bracket{s'}\times\{d'\}$ whose realisation send the $N$-cube $\bigtimes_{i=1}^{N}[k_i/k,(k_{i}+1)/k]$ onto the $N$-cube $\bigtimes_{i=1}^{N}[k'_i/k,(k'_{i}+1)/k]$. There is an edge in $\bar{W}_{\word{w}}$ of source $(s,(k_{i})_{i=1}^{N},d)$ to $(s',(k'_{i})_{i=1}^{N},d')$ if and only if $d=d'$, $k_i=k'_i$ for $i\geqslant 2$ and there is an edge in $W_{\word{w}}$ of source $\bracket{s}\times[k_1/k,(k_1+1)/k]\times[0,1]^{\naturalN}$ and target $\bracket{s'}\times[k'_1/k,(k'_1+1)/k]\times[0,1]^{\naturalN}$.

Then, checking the existence of an alternating path between $M$ and $\oc \graphingtw{w}$ turns out to be equivalent to the existence of an alternating path between $\bar{M}$ and $\graphtw{w}$.
\end{proof}

This lemma will be useful because of the following proposition.

\begin{proposition}\label{orthogonalityandcycles}
For any $\microcosm{m}$-machine $G$ and word representation $\oc W$, $G\plug\oc W$ is orthogonal to $\testdetneg$ if and only if there are no cycles between $G$ and $\oc W\otimes \identity[\support{\ttrm{r}}]$ going through $\support{\ttrm{r}}$.
\end{proposition}

\begin{proof}
We use here the \emph{trefoil property} for graphings \cite{seiller-goig}, which in this case translates as $\meas{(0,G)\plug(0,\oc W),(t,T)}=\meas{(0,G),(0,\oc W)\plug(t,T)}$. Since the support of $\oc W$ and the test are disjoint, we have the equality $(0,\oc W)\plug(t,T)=(0,\oc W)\otimes(t,T)$. Hence $G\plug\oc W$ is orthogonal to $\testdetneg$ if and only if $G$ is orthogonal to $(0,\oc W)\otimes (\zeta,\identity[\support{\ttrm{r}}])$ for all $\zeta\neq 0$. But $(0,\oc W)\otimes (\zeta,\identity[\support{\ttrm{r}}])=(\zeta,\oc W\otimes \identity[\support{\ttrm{r}}])$. Thus $G\plug\oc W$ is orthogonal to $\testdetneg$ if and only if $\zeta+\meas{G,\oc W\otimes \identity[\support{\ttrm{r}}]}\neq 0,\infty$, i.e. if and only if there are no alternating cycles between $G$ and $\oc W\otimes \identity[\support{\ttrm{r}}]$ of weight of the form $a\cdot \mathbf{1}$. Finally, since all weights in $G$ and $\oc W$ are equal to $1$, such cycles need to go through $\support{\ttrm{r}}$.
\end{proof}

Using these results, we can show the wanted inclusion (i.e. completeness of the model). For this we consider a $\Gamma_{N}$-essential $\microcosm{m}_{N}$-machine $G$  where $\Gamma_{N}$ is the subset of $\microcosm{m}_{N}$ in which all permutation-induced transformations are of the form $\ttrm{p}_{\tau_{i,j}}$ where $\tau_{i,j}$ denotes the transposition exchanging $1$ and $j$. We then construct an automaton $\graphauto{G}$ that computes the same language as $G$. We will build the automaton so that it follows the alternating paths between $M$ and $\oc W\otimes\identity[\support{\ttrm{r}}]$ starting in $\support{\ttrm{r}}$, using the fact that this can be done by following the paths between finite graphs $\bar{M}$ and $\bar{W}_{\word{w}}\otimes\identity$ using \autoref{technicallemma}.

We construct the automaton $\graphauto{G}$ as follows. Let $Q$ denote the dialect of the thick graphing $M$. We denote by $\mathfrak{I}$ the set of vertices $(\support{\ttrm{r}},q)$, with $q\in Q$, which are both a source and a target of edges in $M$. Any cycle going through $\support{\ttrm{r}}$ will go through at least one element of $\mathfrak{I}$. Notice however, that such a cycle may go through several elements of $\mathfrak{I}$, i.e. the cycle may go through the test several times before reaching its initial vertex.

If $\mathfrak{I}$ is empty, then $\mathcal{L}^{\testfont{T}}(M)=\emptyset$ which is clearly computed by an automaton with at most $N$ heads. We now suppose that $\mathfrak{I}\neq\emptyset$. We will build an automaton $\graphauto{G}$ whose set of states is equal to $Q\times\mathfrak{G}_{N}\times\mathfrak{I}\times\{\star,0,1\}^{N}$. The permutations in $\mathfrak{G}_{N}$ will be used to keep track of the exchanges of heads during the computation. The sequences in $\{\star,0,1\}^{N}$ will be used to remember the starting positions of the heads: indeed a cycle has to go back not only to its initial state but to its initial heads' positions as well. 

Notice that the choice of an element of $q$ of the dialect together with a sequence in $\{\star,0,1\}^{N}$ corresponds to the choice of a vertex in the graph $\bar{G}$. Notice also that all edges are realised by a transposition $\ttrm{p}_{\tau_{1,j}}$ composed with a bijection on $\naturalN$; we abusively say that the edge is realised by the transposition to lighten the definition of the automaton.

%An edge $e$ in $M$, of source $(\support{\ttrm{r}},q)$ with $q\in\mathcal{I}$ and target $(\support{(s',d')},q')$ realised by $\tau$ is represented by the family of transitions $(\vec{a},(q,\sigma,P,\vec{a}))\rightarrow(\sigma(j),d',(q',\tau_{1,j}\circ\sigma,P,P'\cup\{q\}))$ for $\vec{a}$ such that $a_{\sigma(j)}=s'$ and $P'$ not containing $q$.
\noindent We now define the transition relation of the automaton.
\begin{itemize}[nolistsep,noitemsep]
\item Each edge $e$ in $M$, of source $(\support{\ttrm{r}},q)$ with $q\in\mathcal{I}$ and target $(\support{(s',d')},q')$ realised by $\tau_{1,j}$ is represented by the family of transitions 
\[(\vec{a},(q,\sigma,i,\vec{s}))\rightarrow(\sigma(j),d',(q',\tau_{1,j}\circ\sigma,P\cup\{q\},i,\vec{s}))\]
 for $\vec{a}$ such that $a_{\sigma(j)}=s'$.
\item Each edge $e$ in $M$, of source $(\support{(s,d)},q)$ and target $(\support{(s',d')},q')$ realised by $\tau_{1,j}$ is represented by the family of transitions 
\[(\vec{a},(q,\sigma,i,\vec{s})\rightarrow(\sigma(j),d',(q',\tau_{1,j}\circ\sigma,i,\vec{s}))\]
 for all $\vec{a}$ such that $a_{\sigma(j)}=s'$ and $a_{\sigma(1)}=s$.
\item Each edge $e$ in $M$, of source $(\support{(s,d)},q)$ and target $(\support{\ttrm{r}},q)$ with $q\in\mathfrak{I}$ realised by $\tau_{1,j}$ is represented by:
\begin{itemize}[nolistsep,noitemsep]
\item the family of transitions $(\vec{a},(q,\sigma,i,\vec{s}))\rightarrow(\sigma(j),d',(q',\tau_{1,j}\circ\sigma,i,\vec{s}))$ for $\vec{a}$ such that $a_{\sigma(1)}=s$ and $a_{\sigma(j)}=s'$ and $q\neq i$;
\item the family of transitions $(\vec{a},(q,\sigma,i,\vec{s}))\rightarrow\reject$ for $\vec{a}=\vec{s}$ and $i=q$;
\end{itemize}
\item For each $i\in\mathfrak{I}$ and $\vec{s}\in\{\star,0,1\}^{N}$, there is a transition $(\vec{a},\init)\rightarrow (\vec{a},(i,\identity,i,\vec{a}))$.
\end{itemize}

\begin{definition}
For all integer $N$ and $\Gamma_{N}$-essential $\microcosm{m}_{N}$-machine $G$, we denote $\graphauto{G}$ the $N$-head automaton described above. %The set of states of $\graphauto{G}$ is equal to 
%\[Q\times\mathfrak{G}_{N}\times\mathfrak{Q}\times\{\star,0,1\}^{N}\] 
\end{definition}

The reader can convince herself it is a consequence of the definition of $\autograph{M}$ that, given a word $\word{w}$ as input, it follows nondeterministically all alternating paths between $\bar{G}$ and $\graphtw{w}\otimes\identity[\bar{\ttrm{r}}]$ where $\bar{\ttrm{r}}=\{\ttrm{r}\}\times\{\bigtimes_{i=1}^{N}[k_i/k,(k_{i}+1)/k]~|~0\leqslant k_i\leqslant k-1\}$. From this fact and the fact that such a cycle has to go through one of the vertices in $\bar{\ttrm{r}}$, we obtain the following proposition.

\begin{proposition}\label{cyclestoautomaton}
Let $G$ be a $\Gamma_{N}$-essential $\microcosm{m}_{N}$-machine, $\word{w}$ a word and $\oc W$ be a word representation of $\word{w}$. There is an alternating cycle between $G\plug \oc W$ and $\identity[\support{\ttrm{r}}]$ going through $\support{\ttrm{r}}$ if and only if the automaton $\graphauto{G}$ rejects when given $\word{w}$ as input.
\end{proposition}

%\begin{proposition}\label{acyclictoautomata}
%Let $G$ be an essential $(\{1\},\microcosm{m}_{N})$-machine w.r.t. the generating set $\Gamma$. If $G$ is acyclic, so is $\graphauto{G}$.
%\end{proposition}

\begin{theorem}\label{completeness}
Any language computed by a $\microcosm{m}_{N}$-machine w.r.t. $\testdetneg$ is computed by a deterministic $N$-head automaton.
\end{theorem}

\begin{proof}
The proof consists in combining previous statements. Let $G$ be a $\microcosm{m}_{N}$-machine. Then there exists a $\Gamma_{N}$-essential $\microcosm{m}_{N}$-machine $H$ such that $\mathcal{L}^{\testdetneg}(G)=\mathcal{L}^{\testdetneg}(H)$. Now, we have defined the automaton $\autograph{H}$ which, by \autoref{cyclestoautomaton}, rejects an input $\word{w}$ if and only if there is an alternating path between $H$ and $\graphingtw{w}\otimes\identity[\support{\ttrm{r}}]$ going through $\support{\ttrm{r}}$. But this is equivalent, by \autoref{orthogonalityandcycles}, to the fact that $H\plug\oc\graphingtw{w}$ is not orthogonal to $\testdetneg$. Summing up, we have shown that $\autograph{M}$ rejects $\word{w}$ if and only if $\word{w}\not\in\mathcal{L}^{\testdetneg}(G)$.
\end{proof}

%\section{A Logical Point of View}
%
%\subsection{Interpreting Linear Logic}
%
%
%\begin{definition}[Orthogonality]
%
%\end{definition}
%
%\subsection{External Functoriality}
%
%\begin{prooftree}
%\AxiomC{$\vdash_{i} \oc (A\multimap B)$}
%\AxiomC{$\vdash_{j} \oc (B\multimap C)$}
%\BinaryInfC{$\vdash_{i+j} \oc (A\multimap C)$}
%\end{prooftree}
%
%\begin{prooftree}
%\AxiomC{$\vdash_{i} \Gamma$}
%\UnaryInfC{$\vdash_{i+k} \Gamma$}
%\end{prooftree}

%\begin{prooftree}
%\AxiomC{$X(\star i) \vdash_{1} X(r)$}
%\AxiomC{$X(a) \vdash_{1} X(\star o)$}
%\BinaryInfC{$X(\star o)\multimap X(\star i), X(a) \vdash_{1} X(r)$}
%\AxiomC{$X(0i) \vdash_{1} X(0o)$}
%\UnaryInfC{$\vdash_{1} X(0i) \multimap X(0o)$}
%\AxiomC{$X(1i) \vdash_{1} X(a)$}
%\BinaryInfC{$\vdash_{i+j} \oc (A\multimap C)$}
%\end{prooftree}

\section{Conclusion}

Combining \autoref{soundness} and \autoref{completeness}, we obtain the characterisation of the hierarchy of sublinear complexity classes announced in the introduction.

\begin{theorem}
For all $i\in\naturalN^{\ast}\cup\{\infty\}$, $\pred{m_{i}}=\text{\cctwconfa{i}}$
\end{theorem}

In particular, the microcosm $\microcosm{m}_{1}$ characterises the class of regular languages, while the microcosm $\microcosm{m}_{\infty}$ characterises the class \coNLogspace. 

Future work includes the extension of the techniques to other complexity classes. A similar characterisation of the class of polynomial time predicates should be easily obtained following the recent result by the author and coauthors \cite{seiller-ptime}. This should lead to \Ptime and not \coNPtime since the characterisation is based on pushdown automata \cite{cookP}. Following the syntactic characterisation obtained by Baillot \cite{baillot} by interpreting (some) Turing machines as \ELL proofs, one can expect a characterisation of the nondeterministic polynomial time class \coNPtime. As explained in an overview and perspective paper \cite{seiller-towards}, the results will be adapted for deterministic and probabilistic classes.

\subsection{The Logical View} 

As explained above, the set $\pred{m_{i}}$ can be understood both as semantic restrictions over the set of computable predicates in the model of Elementary Linear Logic described by the microcosm $\microcosm{p}$ (\autoref{theoremELL}), or as the set of computable predicates in a model of a modified linear logic lying in between \MALL and \ELL. Future work in this direction includes the understanding of these intermediate logics, and how they can be described syntactically. Let us provide here a first intuition in this regard. One should notice that functorial promotion is implemented by two steps: the first step uses permutations to prevent the interaction of the information encoded in $[0,1]$ during exponentiation; a second step then takes the two copies of $[0,1]$ and encodes them into a single one by using the function $[\cdot,\cdot]$, obtaining a graphing in the image of the exponentiation operation. The microcosms considered here are obtained by removing the latter function, hence preventing this second step. As a consequence, the models allow for limited composition of exponentiated maps: each new composition requires the use of a new copy of $[0,1]$, and disallow to view those as exponentiated objects themselves. As a consequence, the intuition is that the characterisation of \coNLogspace obtained above corresponds to a restriction of linear logic where arbitrary compositions of exponentiated objects is possible but the resulting object cannot be seen as an exponential object. In some manner, the corresponding system should allow for \emph{external} functorial promotion, in the same sense that countable models of set theory allow for \emph{external} bijections between any two sets regardless of their cardinality in the model.

\subsection{The Geometric View} 

As explained in the introduction and not developed yet, the results we obtain are of a geometrical nature. Indeed, the class are here characterised by microcosms which are (in this case) actions of groups on a measured space. Indeed, the microcosm $\microcosm{m}_{i}$ is obtained from the set of translations on $\integerN$ together with the set of maps induced by the action of the set $\mathfrak{G}_{i}$ of permutations over $\{1,\dots,i\}$ onto the space $[0,1]^{i}$. One should notice that the translations will always exist in any other characterisation of complexity classes using the techniques explained in this paper: this is because they are need to interact with the input. Therefore, only the action of the group $\mathfrak{G}_{i}$ is of importance here. Future work will therefore consider how these group actions are related to the characterisations. Since the integer representation is independent from the group action, it is not difficult to convince oneself that, on one hand, any equivalent -- homotopic -- transformation of the space will give rise to the same complexity class. On the other hand, the group actions considered in this paper can be shown to be non-homotopic by using mathematical invariants \cite{gaboriaucost}. Together with the separation result (\autoref{thm:monien}), this lead the author to the conjecture that the converse holds \cite{seiller-towards}, i.e. that non-equivalent group actions yield distinct complexity classes.

%This shows a correspondence between the group actions corresponding to $\microcosm{m}_{i}$ and the classes ...

%\subparagraph*{Acknowledgements}
%
%I want to thank \dots

%\appendix

\bibliographystyle{abbrv}
\bibliography{current-biblio}

%\received{}{}{}

%\elecappendix

\end{document}